\documentclass[11pt]{article}

\usepackage[margin=1in,headheight=14pt,footskip=30pt]{geometry}
\usepackage{amsmath,amssymb,amsthm,mathtools}
\usepackage{array}
\usepackage{microtype}
\usepackage[hidelinks]{hyperref}
\hypersetup{
  pdftitle={Step Recursion: Mixed Stride Spectra, Path Factorization, and Synchronization Geometry},
  pdfauthor={Kirill Osipov},
  pdfsubject={Mixed stride spectra, path factorization, lattice geometry, higher-rank canonical descent algebras, and synchronization collapse in bounded step recursion},
  pdfkeywords={step recursion, bounded recursion, stride spectrum, Grzegorczyk hierarchy, numerical semigroup, factorization, additive monoid, mixed-stride recursion, vector descent, multi-clock recursion, synchronization collapse, rank-two phase transition}
}

\newcommand{\Nat}{\mathbb N}
\newcommand{\Npos}{\mathbb N_{>0}}
\newcommand{\Hmix}[3]{\mathcal H^{#1}_{#2}[#3]}
\newcommand{\Mprim}[3]{\mathcal M^{#1}_{#2}[#3]}
\newcommand{\Cl}{\operatorname{Cl}}
\newcommand{\SR}{\operatorname{SR}}
\newcommand{\gen}[1]{\langle #1\rangle}
\newcommand{\Sub}{\operatorname{SubMon}}

\newcommand{\Str}{\operatorname{Str}}
\newcommand{\OpStr}{\operatorname{OpStr}}
\newcommand{\VStr}{\mathbf{Str}}
\newcommand{\VSR}{\mathbf{SR}}
\newcommand{\Flat}{\operatorname{flat}}
\newcommand{\Sc}{\operatorname{sc}}
\newcommand{\Rect}{\operatorname{Rect}}
\newcommand{\Atoms}{\operatorname{At}}
\newcommand{\PF}{\operatorname{PF}}
\newcommand{\Len}{\mathsf L}
\newcommand{\sem}[1]{\mathopen{\lbrack\!\lbrack}#1\mathclose{\rbrack\!\rbrack}}

\newtheorem{theorem}{Theorem}
\newtheorem{corollary}[theorem]{Corollary}
\newtheorem{lemma}[theorem]{Lemma}
\newtheorem{proposition}[theorem]{Proposition}
\newtheorem{definition}[theorem]{Definition}
\newtheorem{remark}[theorem]{Remark}

\title{Step Recursion:\\Mixed Stride Spectra, Path Factorization, and Synchronization Geometry}
\author{Kirill Osipov\\
{\normalsize Independent researcher, The Hague, The Netherlands}\\
{\normalsize \href{mailto:d503@acmer.me}{\texttt{d503@acmer.me}}}}
\date{August 2026}

\begin{document}
\maketitle

\begin{abstract}
Step recursion allows recursive computation to move through a canonical
hierarchy in jumps, or \emph{strides}.  Suppose a function algebra is allowed
to use several primitive stride lengths $L$.  Composition immediately produces
sums of these lengths, but it is not clear whether arbitrary nesting of mixed
recursions can create any genuinely new canonical stride.  We prove that it
cannot: at every fixed canonical row $n\ge2$ and lower basis $m<n$, the
canonical strides definable from $L$ are exactly the additive monoid
$\langle L\rangle$ generated by $L$.

The proof gives more information than membership alone.  Every sufficiently
high dependency path computing a canonical descent of stride $p$ carries total
label weight exactly $p$, and the possible path signatures are precisely the
additive factorizations of $p$ by the primitive stride labels.  Thus the
computation remembers both which strides are definable and how each one can be
assembled.  After saturation, inclusion between mixed-stride classes is
exactly inclusion between their additive stride monoids, giving a concrete
lattice description and finite certificates for inclusion.

We then study several synchronized recursion clocks.  Joint descent maps can
encode arbitrary additive submonoids of $\mathbb N^r$, already giving
continuum order complexity for $r=2$.  Ordinary scalar observation, however,
forgets correlations between the clocks and retains only the independent
coordinate strides.  This identifies precisely where synchronization
information is preserved and where it collapses.
\end{abstract}



\section{Introduction}

Step recursion modifies the traversal schedule of bounded recursion while
retaining its update format.  At a fixed canonical row $n$, write
$\Str_n(C)=\{0\}\cup\{p\ge1:\rho_{n,p}\in C\}$.  Since
$\rho_{n,p}\circ\rho_{n,q}=\rho_{n,p+q}$, primitive labels
$L\subseteq\Npos$ generate at least $\gen L$.  The central question is whether
lower-basis computation and arbitrarily nested mixed recursion can create an
additional canonical stride.

For a single fixed label $q$, \cite{Osipov2026Step} gives the principal
spectrum $q\Nat$, reverse divisibility, and the fixed-$q$ canonical-zone
selected-chain machinery used below.  The new difficulty is that a mixed
derivation has no global label: different recursion nodes may contribute
different strides with input-dependent multiplicities, so displacement must be
recovered from the actual occurrence path.  Our main invariant resolves this by a conservation principle for canonical
displacement: lower-basis dependencies may change only the controlled zone
width, whereas each labelled descent contributes its label exactly to the zone
index.  Consequently every sufficiently high selected dependency path computing
$\rho_{n,p}$ has total label weight exactly $p$, and therefore
\[
 \boxed{\Str_n(\Mprim{m}{n}{L})=\gen L}.
\]
More strongly, if $A=\{a_1,\ldots,a_k\}$ is the intrinsic atom set of the
recovered monoid $\Gamma$, then the path signatures of $\rho_{n,p}$ are
exactly the solutions of $a_1z_1+\cdots+a_kz_k=p$.  Thus the computation
recovers not only spectrum membership but the complete additive
factorization fibre; finite presentations of $\Gamma$ become finite rewrite
bases for realizable path signatures.

Saturation by all strides of $\Gamma$ makes the index intrinsic both
extensionally and operationally:
\[
 \OpStr_n(\Hmix{m}{n}{\Gamma})
 =\Str_n(\Hmix{m}{n}{\Gamma})=\Gamma,
 \qquad
 \Hmix{m}{n}{\Gamma}\subseteq\Hmix{m}{n}{\Lambda}
 \Longleftrightarrow\Gamma\subseteq\Lambda.
\]
Hence the saturated mixed sector is lattice-isomorphic to $\Sub(\Nat)$.
This computational classification transports the standard finite-generation,
oversemigroup, and interval geometry of submonoids of $\Nat$ to the saturated
Step Recursion sector.  We record the resulting principal skeleton and finite
inclusion certificates, while treating the underlying one-dimensional
semigroup facts as classical consequences rather than independent novelty.

Higher rank separates joint address geometry from its observable scalar
shadow.  Parallel canonical descents, as maps $\Nat^r\to\Nat^r$, realize
$\Sub(\Nat^r)$; already for $r=2$ the powerset order embeds, yielding continuum
many descent monoids and infinitely generated examples.  Prime-exponent address
coding gives a faithful unary conjugacy of those maps, so scalarity of the
\emph{representation} is not the obstruction; the coding itself is not claimed
to belong to the Step Recursion class.  By contrast, synchronized multi-clock
recursion viewed through ordinary canonical scalar membership retains only the
monoid $M(L)$ generated by the individual coordinates of its labels:
\[
 \boxed{\Str_n(\mathcal M^{m,(r)}_n[L])=M(L)},\qquad
 \boxed{\VStr_{n,r}(\mathcal M^{m,(r)}_n[L])=M(L)^r}.
\]
For $\Gamma\le\Nat^r$, this loss is the reflection
$\Gamma\mapsto\Rect_r(\Gamma)=\Sc(\Gamma)^r$ onto rectangular monoids;
every nonzero fibre contains a copy of $\mathcal P(\Nat)$.  Thus the
obstruction is separable observation rather than scalar representation itself.

The algebraic ingredients themselves are classical.  Changes of recursion
schemes and lattice phenomena in subrecursive hierarchies go back at least to
\cite{Axt1966,Machtey1971,Machtey1974,Nielsen2022}; numerical-semigroup
factorization and presentations are standard
\cite{RosalesGarciaSanchez2009,ONeillPelayo2017}, and the oversemigroups of a
fixed numerical semigroup are studied explicitly in
\cite{RosalesOversemigroups2003}.  Higher-rank submonoids are studied in
\cite{Gotti2020}, simultaneous-recursion coding has a separate
literature \cite{Clote1999,Xirotiri2006}, and unary/function-algebra reductions
have a substantial history \cite{Severin2008,Szalkai1985,LambekScott2005}.
The novelty claimed here is narrower: variable generalized-inverse labels are
tracked along actual occurrence-respecting dependency paths, giving exact
weighted displacement and exact factorization fibres at a fixed lower basis
and canonical row.  The multi-clock results concern synchronized recursion
\emph{addresses}, not several simultaneously defined outputs; the coding
result is used as a conjugacy, not claimed as a new coding technique.

\section{Canonical step recursion}
\label{sec:framework}

Throughout,
\[
 \Nat=\{0,1,2,\ldots\},\qquad \Npos=\Nat\setminus\{0\}.
\]
For a unary map $u$, write $u^{[0]}(x)=x$ and
$u^{[t+1]}(x)=u(u^{[t]}(x))$.

We use Rose's level numbering for the Grzegorczyk bases
\cite{Grzegorczyk1953,Rose1984}, with generating functions indexed from one:
\[
 e_1(x,y)=x+y,\qquad e_2(x)=x^2+2,
\]
and, for $r\ge2$,
\[
 e_{r+1}(0)=2,\qquad e_{r+1}(x+1)=e_r(e_{r+1}(x)).
\]
Thus $e_{r+1}(x)=e_r^{[x]}(2)$ for $r\ge2$.  Let
\[
 B_m=\{Z,S\}\cup\{P_i^k:k\ge1,\ 1\le i\le k\}
      \cup\{e_j:1\le j\le m\}.
\]
Only canonical rows $n\ge2$ are needed, and we put $g_n=e_n$.  For $l\ge1$,
let $\rho_{n,l}$ be the generalized inverse of $g_n^{[l]}$:
\[
 \rho_{n,l}(0)=0,\qquad
 \rho_{n,l}(y)=\min\{z\in\Nat:g_n^{[l]}(z)\ge y\}\quad(y>0).
\]
The corresponding descent depth is
\[
 D_{n,l}(y)=\min\{t:\rho_{n,l}^{[t]}(y)=0\}.
\]
For algebraic statements involving the zero stride, we use the harmless
convention $\rho_{n,0}=\operatorname{id}_{\Nat}$.

\begin{definition}[bounded step recursion]
Let $C$ be a family of total functions and let $\rho$ be a descent.  If earlier
functions
\[
 g:\Nat^k\to\Nat,\qquad
 h:\Nat^{k+2}\to\Nat,\qquad
 b:\Nat^{k+1}\to\Nat
\]
belong to the current closure, bounded step recursion along $\rho$ forms
$f:\Nat^{k+1}\to\Nat$ by
\[
 f(\bar x,0)=g(\bar x),\qquad
 f(\bar x,y)=h\bigl(\bar x,\rho(y),f(\bar x,\rho(y))\bigr)
\]
for $y>0$, provided $f(\bar x,y)\le b(\bar x,y)$ everywhere.  We denote the
operation by $\SR_\rho$.
\end{definition}

Two elementary facts will be used repeatedly.

\begin{lemma}[canonical stride arithmetic]
\label{lem:stride-arithmetic}
For $n\ge2$ and $a,b\ge1$,
\[
 \rho_{n,a+b}=\rho_{n,a}\circ\rho_{n,b}
              =\rho_{n,b}\circ\rho_{n,a}.
\]
More generally, if $p=kl$, then
$\rho_{n,p}=\rho_{n,l}^{[k]}$.
\end{lemma}

\begin{proof}
Let $\rho=\rho_{n,1}$.  For a strictly increasing integer-valued map, the
generalized inverse of a finite iterate is the corresponding iterate of the
generalized inverse.  Hence $\rho_{n,l}=\rho^{[l]}$, and the assertions are
ordinary laws of iteration.
\end{proof}

\begin{lemma}[an admitted descent is internal]
\label{lem:descent-internal}
Let $C$ contain zero and the projections and be closed under bounded step
recursion along a descent $\rho$.  Then $\rho\in C$.
\end{lemma}

\begin{proof}
Use one step recursion along $\rho$ with base $0$ and transition equal to the
current recursion address:
\[
 f(0)=0,\qquad f(y)=\rho(y)\quad(y>0).
\]
The identity function is an admissible bound.
\end{proof}

For canonical inputs put
\[
 Y_t=g_n^{[t]}(0)\qquad(t\in\Nat).
\]
Then for $t\ge l$,
\begin{equation}
\label{eq:canonical-shift}
 \rho_{n,l}(Y_t)=Y_{t-l}.
\end{equation}
The sequence $(Y_t)$ is the layer coordinate on which the separation proof is
performed.  We shall also use the elementary generalized-inverse threshold law
\begin{equation}
\label{eq:canonical-threshold}
 D_{n,1}(u)\le k
 \quad\Longleftrightarrow\quad
 u\le Y_k.
\end{equation}
Indeed, $\rho_{n,1}$ is the generalized inverse of $G$, so the inputs reaching
zero in at most $k$ predecessor steps are exactly the initial segment ending at
$G^{[k]}(0)=Y_k$.

\section{Mixed strides and additive monoids}

\begin{definition}[canonical stride spectrum]
\label{def:stride-spectrum}
For any function class $C$ and fixed canonical row $n\ge2$, define
\[
 \boxed{\Str_n(C):=\{0\}\cup\{p\in\Npos:\rho_{n,p}\in C\}.}
\]
The spectrum records only canonical descent functions already definable in
$C$; it does not assume that $C$ is closed under recursion along every descent
whose label lies in the spectrum.
\end{definition}

\begin{definition}[primitive-family class]
\label{def:primitive-family}
For any set $L\subseteq\Npos$ of admitted primitive stride labels, define
\[
 \boxed{
 \Mprim{m}{n}{L}
 :=\Cl_{\circ,\{\SR_{\rho_{n,l}}:l\in L\}}(B_m).}
\]
If $L=\varnothing$, this is the composition closure $\Cl_\circ(B_m)$.
Write $\gen{L}$ for the additive submonoid of $\Nat$ generated by $L$, with
$0\in\gen{L}$.
\end{definition}

\begin{definition}[mixed-stride class]
\label{def:mixed}
Let $\Gamma\le(\Nat,+)$ be an additive submonoid.  Define
\[
 \boxed{
 \Hmix{m}{n}{\Gamma}
 :=\Mprim{m}{n}{\Gamma\setminus\{0\}}
 =\Cl_{\circ,\{\SR_{\rho_{n,l}}:l\in\Gamma\setminus\{0\}\}}(B_m).}
\]
If $\Gamma=\{0\}$, this means the composition closure
$\Cl_\circ(B_m)$.
\end{definition}

Every nonzero additive submonoid has a finite canonical basis.

\begin{proposition}[finite canonical atom basis]
\label{lem:finite-generation}
\label{prop:spectrum-basis}
\label{cor:finite-spectrum-basis}
Every nonzero $\Gamma\le(\Nat,+)$ is finitely generated, and its atoms
\[
 \Atoms(\Gamma)=\{a\in\Gamma\setminus\{0\}:a\ne b+c
 \text{ for }b,c>0\}
\]
form its unique minimal generating set.  Hence for every possibly infinite
$L\subseteq\Npos$, if $\Gamma=\gen L$ and $A=\Atoms(\Gamma)$, then $A$ is
finite, $A\subseteq L$, and $\gen A=\gen L$.
\end{proposition}

\begin{proof}
Let $d=\min(\Gamma\setminus\{0\})$.  For each residue modulo $d$ occurring
in $\Gamma$, choose its least representative $a_r$; then every
$x\in\Gamma$ is $a_r+kd$, proving finite generation.  An inclusion-minimal
finite generating set consists exactly of the atoms, and every generating set
must contain every atom.  If $a\in\Atoms(\gen L)$, an expression of $a$ as a
sum of elements of $L$ has length one, so $a\in L$.
\end{proof}

If $d=\gcd(\Gamma\setminus\{0\})$, then $d^{-1}\Gamma$ has gcd one and is a
numerical semigroup \cite{RosalesGarciaSanchez2009}.

\section{The mixed selected-chain invariant}
\label{sec:mixed-invariant}

We now prove the reverse inclusion: below $m<n$, mixed recursion recovers no
canonical stride outside the generated monoid.  Fix
\[
 n\ge2,\qquad m<n,
\]
and put
\[
 G=g_n=e_n,\qquad Y_t=G^{[t]}(0).
\]
Use the lower-row majorant
\[
 F(x)=2x+2\quad(n=2),\qquad F=e_{n-1}\quad(n\ge3).
\]
For integers $s,p\ge0$ define
\[
 U_s[p]=F^{[p]}(Y_s+p),\qquad
 L_s[p]=\min\{u:F^{[p]}(u+p)\ge Y_s\},
\]
and the width-$p$ canonical zone
\[
 \mathcal Z_s[p]=\{u:L_s[p]\le u\le U_s[p]\}.
\]
If $P$ is a polynomial with nonnegative integer coefficients, abbreviate
\[
 \mathcal Z_{s,P}(t)=\mathcal Z_s[P(t)].
\]
The point $Y_s$ lies in $\mathcal Z_s[0]$.

The growth estimate needed for zone separation is proved here.

\begin{lemma}[quantitative canonical gap]
\label{lem:canonical-gap}
For every polynomial $P$ with nonnegative integer coefficients and every fixed
integer offset $a$, whenever the displayed indices are nonnegative,
\begin{equation}
\label{eq:canonical-gap}
 F^{[P(t)]}\bigl(Y_{t-a-1}+P(t)\bigr)<Y_{t-a}
\end{equation}
for all sufficiently large $t$.
\end{lemma}

\begin{proof}
Put $s=t-a-1$.  Since $G=e_n\ge e_2$ and $e_2(x)=x^2+2$, induction gives
$Y_j\ge 2^{2^{j-1}}$ for all sufficiently large $j$.  Hence $Y_j$
eventually dominates every fixed polynomial in $j$, and even
$2^{P(j)}P(j)$ for every fixed polynomial $P$.

If $n\ge3$, then $F=e_{n-1}$ and the defining identities give
\[
 Y_s=F^{[Y_{s-1}]}(2),\qquad Y_{s+1}=F^{[Y_s]}(2).
\]
For large $t$, $F(Y_s)>Y_s+P(t)$, hence
\[
 F^{[P(t)]}(Y_s+P(t))
 <F^{[P(t)+1]}(Y_s)
 =F^{[Y_{s-1}+P(t)+1]}(2).
\]
Also $Y_s-Y_{s-1}>P(t)+1$ eventually, so the last value is strictly below
$F^{[Y_s]}(2)=Y_{s+1}$.

If $n=2$, then $G(x)=x^2+2$ and $F(x)=2x+2$, whence
\[
 F^{[r]}(x)=2^r(x+2)-2,\qquad Y_{s+1}=Y_s^2+2.
\]
Since $Y_s$ eventually dominates $2^{P(t)}(P(t)+2)$,
\[
 2^{P(t)}(Y_s+P(t)+2)-2<Y_s^2+2=Y_{s+1}.
\]
Replacing $s$ by $t-a-1$ proves \eqref{eq:canonical-gap}.
\end{proof}

This estimate implies the following separation fact.

\begin{lemma}[zone uniqueness]
\label{lem:zone-unique}
For every fixed $d$ and polynomial $P$, for all sufficiently large $t$,
\[
 Y_{t-d}\in\mathcal Z_{s,P}(t)\quad\Longrightarrow\quad s=t-d.
\]
\end{lemma}

\begin{proof}
If $s\le t-d-1$, then
\[
 Y_{t-d}\le U_s[P(t)]
 \le F^{[P(t)]}\bigl(Y_{t-d-1}+P(t)\bigr)<Y_{t-d}
\]
by~\eqref{eq:canonical-gap}, a contradiction.  If $s\ge t-d+1$, the lower-zone
inequality gives
\[
 Y_{t-d+1}\le Y_s
 \le F^{[P(t)]}\bigl(Y_{t-d}+P(t)\bigr)<Y_{t-d+1},
\]
again a contradiction.  Hence $s=t-d$.
\end{proof}

\begin{lemma}[lower-basis majorization and oriented primitive dependencies]
\label{lem:lower-majorant}
Fix $n\ge2$ and $m<n$, and put
\[
 F(x)=2x+2\quad(n=2),\qquad F=e_{n-1}\quad(n\ge3).
\]
For every fixed composition term over $B_m$ there is $c\ge1$ such that, if all
its arguments are at most $u$, then every value produced by that term is at
most $F^{[c]}(u+c)$.  At a concrete nonconstant initial-function occurrence one
can moreover choose a single oriented dependency $a\mapsto v$ with $a\le v$:
for successor choose its input; for addition $e_1(x,y)=x+y$ choose an input of
maximum value; for a unary $e_j$ choose its unique input; and for a projection
choose the projected input, giving equality.  Zero has no input dependency and
is treated as a fixed origin.
\end{lemma}

\begin{proof}
For $n=2$ one has $m\le1$, so the only nontrivial growing initial functions are
successor and, when $m=1$, addition; both are bounded by a fixed iterate of
$F(x)=2x+2$.  For $n\ge3$, every unary generator available in $B_m$ is bounded
by $F=e_{n-1}$ after at most a fixed iterate, and $F(x)\ge2x$ absorbs successor
and addition.  Structural induction on a fixed composition term gives the
uniform envelope $F^{[c]}(u+c)$.

The oriented dependency rules are explicit.  Successor and every unary
$e_j$ are nondecreasing and dominate their selected input; addition dominates
its larger input; a projection reproduces its selected input exactly.  These
are all nonconstant initial functions in $B_m$.  Enlarging $c$ once handles the
finitely many primitive occurrences in any fixed term.
\end{proof}

The next two lemmas are deliberately formulated with an arbitrary integer width
$p$.  This avoids any dependence of the zone invariant on how many basis steps
have already occurred on a particular selected chain.

\begin{lemma}[uniform basis-step stability]
\label{lem:mixed-basis-stability}
For every fixed mixed derivation $\delta$ there is a constant $C_\delta\ge1$
such that every selected lower-basis step $u\mapsto v$ satisfies
\[
 u\le v\le F^{[C_\delta]}(u+C_\delta).
\]
Consequently, for every integer $p\ge0$,
\[
 u\in\mathcal Z_s[p]
 \quad\Longrightarrow\quad
 v\in\mathcal Z_s[p+C_\delta].
\]
\end{lemma}

\begin{proof}
Apply Lemma~\ref{lem:lower-majorant} to the finitely many lower-basis
composition fragments occurring in the fixed derivation $\delta$ and take the
maximum of their constants.

For the zone statement, monotonicity of $F$ gives
\[
 v\le F^{[C_\delta]}(u+C_\delta)
 \le F^{[p+C_\delta]}(Y_s+p+C_\delta),
\]
while $v\ge u$ gives
\[
 Y_s\le F^{[p]}(u+p)
 \le F^{[p+C_\delta]}(v+p+C_\delta).
\]
These are the two defining inequalities for
$\mathcal Z_s[p+C_\delta]$.
\end{proof}

The next observation is uniform in the stride label.  It is the reason that
the layer valuation is additive in the labels rather than controlled only by
their maximum.

\begin{lemma}[uniform descent displacement]
\label{lem:mixed-descent-displacement}
For every $q\ge1$, every integer $p\ge0$, every $s\ge q$, and every $u$,
\[
 u\in\mathcal Z_s[p]
 \quad\Longrightarrow\quad
 \rho_{n,q}(u)\in\mathcal Z_{s-q}[p].
\]
\end{lemma}

\begin{proof}
Write $H_p(x)=F^{[p]}(x+p)$.  The canonical generators satisfy
\[
 G(F(x))\ge F(G(x)).
\]
For $n=2$ this is the direct inequality
$(2x+2)^2+2\ge2(x^2+2)+2$.  For $n\ge3$, $F(x)\ge x+1$ and
$G(x+1)=F(G(x))$, so monotonicity of $G$ gives the same inequality.
Induction on $p$ gives
\[
 G\bigl(F^{[p]}(z)\bigr)\ge F^{[p]}(G(z)),
\]
and iteration in $q$ yields
\[
 G^{[q]}\bigl(F^{[p]}(z)\bigr)
 \ge F^{[p]}\bigl(G^{[q]}(z)\bigr).
\]
Set $z=x+p$.  Since $G^{[q]}$ is a strictly increasing integer-valued map,
\[
 G^{[q]}(x+p)\ge G^{[q]}(x)+p.
\]
Therefore
\begin{equation}
\label{eq:mixed-shifted-commutation}
 G^{[q]}(H_p(x))\ge H_p(G^{[q]}(x)).
\end{equation}
Applying~\eqref{eq:mixed-shifted-commutation} at
$x=\rho_{n,q}(w)$ and using
$G^{[q]}(\rho_{n,q}(w))\ge w$ gives
\[
 \rho_{n,q}(H_p(w))\le H_p(\rho_{n,q}(w)).
\]

Now suppose $u\in\mathcal Z_s[p]$ and $s\ge q$.  The upper zone inequality gives
\[
 \rho_{n,q}(u)
 \le \rho_{n,q}(H_p(Y_s))
 \le H_p(\rho_{n,q}(Y_s))
 =H_p(Y_{s-q}).
\]
The lower zone inequality $Y_s\le H_p(u)$ gives
\[
 Y_{s-q}=\rho_{n,q}(Y_s)
 \le\rho_{n,q}(H_p(u))
 \le H_p(\rho_{n,q}(u)).
\]
These are precisely the two defining inequalities for
$\mathcal Z_{s-q}[p]$.
\end{proof}

\subsection*{Mixed selected dependency chains}

A \emph{mixed derivation} is a finite derivation over $B_m$ in which each
step-recursion node carries its own label $q\ge1$ and uses $\rho_{n,q}$.  Its
\emph{construction rank} is the height of the derivation tree: initial leaves have
rank zero, and composition or step-recursion nodes have rank one plus the maximum
rank of their proper subderivations.  Let
\[
 Q_\delta\subseteq\Npos
\]
be the finite set of stride labels occurring in a fixed derivation $\delta$.

\begin{definition}[mixed selected dependency chain]
\label{def:mixed-selected-chain}
Evaluate a mixed derivation $\delta$ at concrete inputs.  A mixed selected
dependency chain is an occurrence-respecting path
\[
 u_0,u_1,\ldots,u_r=\sem{\delta}(\bar x)
\]
whose origin is either a formal input occurrence or a fixed numeral occurrence.
It is constructed inductively.  At an initial-function node use the oriented
primitive dependency from Lemma~\ref{lem:lower-majorant}: successor follows its
input, addition follows a maximum-valued input, a unary generator follows its
unique input, a projection gives equality, and zero is a fixed origin.

At a composition node, first choose a selected chain in the concrete head
evaluation.  If that chain begins at the head's $i$th formal input, prepend a
selected chain through the corresponding argument subevaluation.  At a
step-recursion node labelled $q$ with current recursion address $y=0$, choose a
selected chain in the concrete base subderivation.  If $y>0$, choose a selected
chain in the concrete transition evaluation.  If its origin is a parameter,
link it to the corresponding outer input; if it is the current address,
prepend the relevant segment of the $\rho_{n,q}$ schedule; if it is the
previous state, prepend a selected chain for the preceding recursive state.
Following previous-state origins always terminates, because every such move
decreases the finite stage index along the $\rho_{n,q}$ schedule.  A fixed
numeral remains a fixed origin.  No concrete occurrence is reused.

After suppressing occurrence labels, every transition of the resulting
numerical path is exactly one of:
\begin{enumerate}
\item an equality step;
\item an oriented lower-basis step $u_j\mapsto u_{j+1}$ satisfying
\[
 u_j\le u_{j+1}\le F^{[C]}(u_j+C)
\]
for a derivation-dependent constant $C$;
\item a descent step
$u_{j+1}=\rho_{n,q}(u_j)$ carrying its actual label $q\in Q_\delta$.
\end{enumerate}
Thus a mixed selected chain records one genuine numerical dependency path
through the evaluation together with the ordered word of stride labels met on
that path; it is not the full evaluation tree.
\end{definition}

\begin{lemma}[stagewise lifting and well-foundedness of selected paths]
\label{lem:mixed-path-lifting}
Fix a mixed derivation $\delta$ and a concrete evaluation of one of its output
occurrences.  The selection rules of Definition~\ref{def:mixed-selected-chain}
produce at least one finite occurrence-respecting selected chain ending at that
output, and every chain produced by those rules is finite.  More precisely, at a
step-recursion node labelled $q$, a selected transition dependency at stage $j$
lifts to the enclosing evaluation as follows:
\begin{enumerate}
\item a parameter origin lifts through the corresponding outer argument;
\item an address origin lifts through the concrete predecessor edge carrying the
label $q$;
\item a previous-state origin lifts to a selected chain for stage $j-1$.
\end{enumerate}
Repeated previous-state lifting terminates after at most $j$ such moves.  After
all lifts are expanded, every numerical transition is an equality, an oriented
lower-basis step, or one actual descent $u\mapsto\rho_{n,q}(u)$ carrying the label
of the recursion node at which it occurs.  No evaluation occurrence is visited
twice.
\end{lemma}

\begin{proof}
Use induction on the construction rank of $\delta$.  Initial-function and
composition nodes are immediate from the explicit rules and the induction
hypothesis for proper subderivations.  At a step-recursion node, keep the
construction-rank induction fixed and use a secondary induction on the concrete
stage number $j$.  At stage $0$ the chain lies in the base subderivation.  At a
positive stage, first select a chain in the transition subevaluation.  Parameter
and address origins lift directly as stated.  If the selected origin is the
previous state, invoke the stage induction at $j-1$ and prepend the resulting
chain.  Hence this unwinding decreases $j$ strictly and terminates.  Proper
subderivations have smaller construction rank, so all recursively requested
chains inside the base or transition are finite.  The edge-type assertion follows
from the oriented primitive rules and from recording each concrete schedule edge
at the recursion node that generated it.  Since construction occurrences are
followed downward and stage indices decrease whenever a previous state is
unwound, an occurrence cannot be revisited.
\end{proof}

The proof below therefore uses three separate well-founded arguments: construction
rank guarantees finite subderivation lifting, the stage index guarantees
termination inside a recursion node, and later the position along the resulting
finite chain supports the zone-valuation induction.  Keeping these inductions
separate is useful because the mixed labels themselves are not ordered.

The next lemma supplies the uniform quantitative bound needed for the mixed
proof.  In the evaluation semantics used here, the declared
bound is a static admissibility certificate; evaluating that certificate is not
part of the operational value set of the recursion node.

\begin{lemma}[mixed internal-value and depth envelope]
\label{lem:mixed-internal-envelope}
For every fixed mixed derivation $\delta$ there is a constant $c_\delta\ge1$
with the following property.  If every input coordinate is at most $N$, then
every numerical value created in the operational evaluation of $\delta$ is at
most
\[
 G^{[c_\delta]}(N+c_\delta).
\]
Consequently there is $C_\delta'\ge1$ such that every recursion address $z$
occurring at a node labelled $q$ satisfies
\[
 D_{n,q}(z)
 \le C_\delta'\bigl(1+D_{n,1}(N+C_\delta')\bigr).
\]
In particular, at canonical input $Y_t$ every node-specific recursion depth is
$O_\delta(t+1)$.
\end{lemma}

\begin{proof}
We prove the value bound by induction on the construction rank of $\delta$.
Because $m<n$, every initial function in $B_m$ is bounded, on arguments at most
$N$, by $G^{[c]}(N+c)$ for a fixed $c$ depending only on that initial function;
this is Lemma~\ref{lem:lower-majorant}.  Thus the claim holds for initial
derivations.

At a composition node, the induction hypotheses bound all values created in the
argument subevaluations by one common envelope
\[
 A=G^{[a]}(N+a).
\]
Apply the induction hypothesis for the head derivation with input cap $A$.  If
the head contributes a fixed envelope $G^{[b]}(A+b)$, choose a fixed $r$ with
$2r\ge b$.  Since $G(x)\ge x+2$, one has $A+b\le G^{[r]}(A)$ and hence
\[
 G^{[b]}(A+b)
 \le G^{[b+r]}(A)
 =G^{[a+b+r]}(N+a)
 \le G^{[c]}(N+c)
\]
for one fixed $c$.  Thus finite composition preserves the required envelope
form.

Consider now a step-recursion node labelled $q$, evaluated at $(\bar p,z)$ with
all coordinates at most $N$.  Put $d=D_{n,q}(z)$ and write its schedule as
\[
 z_j=\rho_{n,q}^{[d-j]}(z)\qquad(0\le j\le d).
\]
Every $z_j$ is at most $z\le N$.  Let $\beta$ be the earlier bound derivation.
By the induction hypothesis its output on $(\bar p,z_j)$ is at most
$B=G^{[b]}(N+b)$, uniformly in $j$.  If $s_j$ denotes the recursive state at
$z_j$, admissibility gives
\[
 s_j\le\sem{\beta}(\bar p,z_j)\le B.
\]
The base derivation therefore receives inputs at most $N$, and every transition
evaluation receives parameters and address at most $N$ and previous state at
most $B$.  Applying the induction hypotheses to the base and transition with
input caps $N$ and $\max(N,B)$ bounds every value created in all those
subevaluations by finitely many envelopes of the required form.  The same
absorption argument used in the composition case combines those finitely many
fixed iterates and shifts into one $G^{[c_\delta]}(N+c_\delta)$, independently
of the outer recursion depth $d$.  The induction already covers all nested mixed
labels, so the value bound is uniform over the finite set $Q_\delta$.

Every recursion address is among the values just bounded.  Since
$\rho_{n,q}=\rho_{n,1}^{[q]}$, one has $D_{n,q}(u)\le D_{n,1}(u)$ for all $u$.
Let $r=D_{n,1}(N+c_\delta)$.  The canonical threshold law gives
$N+c_\delta\le G^{[r]}(0)$ and hence
\[
 G^{[c_\delta]}(N+c_\delta)\le G^{[r+c_\delta]}(0).
\]
Thus the fine depth of every address is at most $r+c_\delta$, which gives the
displayed estimate after enlarging the constant.  Finally, \eqref{eq:canonical-threshold} gives $D_{n,1}(Y_t)=t$.  More
explicitly, for every fixed $c\ge0$ choose $h=\lceil c/2\rceil$.  Since
$G(x)\ge x+2$,
\[
 Y_t\le Y_t+c\le G^{[h]}(Y_t)=Y_{t+h},
\]
and monotonicity together with \eqref{eq:canonical-threshold} yields
\[
 t\le D_{n,1}(Y_t+c)\le t+h.
\]
Thus every fixed additive shift changes canonical depth by at most a fixed
amount, and all node-specific depths are $O_\delta(t+1)$ on canonical inputs.
\end{proof}

\begin{lemma}[uniform polynomial length of mixed chains]
\label{lem:mixed-chain-length}
For every fixed mixed derivation $\delta$ there is a polynomial $P_\delta$ and
a finite set of numeral origins $\mathcal N_\delta$ such that, in every
evaluation of $\delta$ at $Y_t$, \emph{every} mixed selected dependency chain
has length at most $P_\delta(t)$.  If such a chain has a numeral origin, that
origin belongs to $\mathcal N_\delta$.
\end{lemma}

\begin{proof}
First suppose every recursion address occurring in an evaluation has
node-specific descent depth at most $K$.  Let $L_\delta(K)$ denote a uniform
upper bound for the lengths of \emph{all} mixed selected dependency chains in
such evaluations.  We construct it by induction on the construction rank of
$\delta$.  Initial functions contribute constant-length chains using the explicit
oriented rules in Definition~\ref{def:mixed-selected-chain}; zero contributes
only a numeral origin.  At a composition node an arbitrary selected chain
follows the head chain and then one selected argument branch, so its length is
bounded by the head-chain bound plus the maximum argument-chain bound.  At a
step-recursion node labelled $q$, an arbitrary selected chain crosses at most
$d=D_{n,q}(z)\le K$ recursion stages and, within each crossed stage, follows at
most one selected transition subchain.  Hence
\[
 L_\delta(K)
 \le L_\gamma(K)+K L_\eta(K)+K+1.
\]
The derivations $\gamma$ and $\eta$ have strictly smaller construction rank, so
the induction hypotheses make $L_\gamma$ and $L_\eta$ polynomial; multiplication
by $K$ preserves polynomiality.  Thus
$L_\delta(K)\le A_\delta(K+1)^{r_\delta}$ for fixed constants, uniformly over
all selected chains.  A fixed derivation has only finitely many numeral
occurrences, giving $\mathcal N_\delta$.  Lemma~\ref{lem:mixed-internal-envelope}
supplies a polynomial bound on $K$ at input $Y_t$; substitution yields a fixed
polynomial $P_\delta(t)$ valid for every selected chain.
\end{proof}

We can now state the genuinely mixed valuation theorem.  If a chain segment
contains descent labels $q_1,\ldots,q_r$, its stride weight is
$q_1+\cdots+q_r$.

\begin{theorem}[mixed zone-index valuation]
\label{thm:mixed-zone-valuation}
Let a mixed selected dependency chain for $\delta$ contain a point
$u_i\in\mathcal Z_s[P(t)]$ for a fixed polynomial $P$.  For $h\ge i$, let
$w_h$ be the sum of the labels of all descent steps between $u_i$ and $u_h$, and
let $b_h$ be the number of lower-basis steps on that segment.  As long as every
descent step labelled $q$ starts when the current zone index is at least $q$,
one has the exact integer-width invariant
\[
 u_h\in\mathcal Z_{s-w_h}\bigl[P(t)+b_hC_\delta\bigr].
\]
Consequently, with
\[
 P^\star(X)=P(X)+C_\delta P_\delta(X),
\]
one has
\[
 u_h\in\mathcal Z_{s-w_h,P^\star}(t)
\]
for every $h$ on the segment.
\end{theorem}

\begin{proof}
Induct along the chain.  Equality leaves both index and width unchanged.  A
basis step preserves the index and increases the current integer width by at
most $C_\delta$ by Lemma~\ref{lem:mixed-basis-stability}.  A descent step
labelled $q$ changes the current index from $s'$ to $s'-q$ without changing the
integer width by Lemma~\ref{lem:mixed-descent-displacement}.  This proves the
exact invariant with width $P(t)+b_hC_\delta$.

By Lemma~\ref{lem:mixed-chain-length}, $b_h\le P_\delta(t)$.  Hence the current
integer width is at most $P^\star(t)$.  Enlarging an integer zone width preserves
membership by monotonicity of $F$, so the polynomial-width conclusion follows.
\end{proof}

Two barriers rule out the only ways in which the valuation could fail near the
bottom of the canonical scale.

\begin{lemma}[constant-origin exclusion]
\label{lem:mixed-constant-origin}
Let $\delta$ be fixed and let $d\ge0$.  For all sufficiently large $t$, no mixed
selected dependency chain for $\delta$ that begins at a fixed numeral can end at
$Y_{t-d}$.
\end{lemma}

\begin{proof}
For an upper envelope, replace every descent step by the identity.  Since each
canonical descent is nonincreasing, this monotone replacement can only enlarge
the propagated majorant.  By Lemmas~\ref{lem:mixed-basis-stability} and
Lemma~\ref{lem:mixed-chain-length}, a chain starting from one of the finitely
many numerals in $\mathcal N_\delta$ is bounded by
\[
 F^{[R(t)]}(c+R(t))
\]
for a fixed polynomial $R$ and fixed constant $c$.  Choose fixed $s_0$ with
$c\le Y_{s_0}$.  For large $t$, $s_0\le t-d-1$, so monotonicity followed by
\eqref{eq:canonical-gap} places this upper bound strictly below $Y_{t-d}$.
\end{proof}

\begin{lemma}[mixed low-zone barrier]
\label{lem:mixed-low-zone}
Let $\delta$ be fixed, let $d\ge0$, and put
$q_{\max}=\max Q_\delta$ when $Q_\delta\ne\varnothing$.  Consider a selected
chain for $\delta(Y_t)$ that starts at the external input
$Y_t\in\mathcal Z_t[0]$.  Suppose all descent steps before a given
$q$-labelled step satisfy the legality hypothesis of
Theorem~\ref{thm:mixed-zone-valuation}, but immediately before that step the
valued zone index $s$ satisfies $s<q$.  Then, for all sufficiently large $t$,
the chain cannot end at $Y_{t-d}$.
\end{lemma}

\begin{proof}
Theorem~\ref{thm:mixed-zone-valuation} applies to the entire prefix before the
offending step, so at that point the current value is at most $U_s[p]$ for an
integer width $p$ bounded by a fixed polynomial in $t$.  Since
$s<q\le q_{\max}$, one has $s\le q_{\max}-1$, independently of $t$.  For a
monotone upper envelope, replace the offending descent and every later
descent by the identity; because all canonical descents are nonincreasing, this
can only enlarge the propagated majorant.  The remaining number of basis steps
is polynomial by
Lemma~\ref{lem:mixed-chain-length}, and repeated use of
Lemma~\ref{lem:mixed-basis-stability} bounds the eventual value by
\[
 F^{[R(t)]}\bigl(Y_{q_{\max}-1}+R(t)\bigr)
\]
for a fixed polynomial $R$.  For large $t$ one has
$q_{\max}-1\le t-d-1$, so monotonicity followed by
\eqref{eq:canonical-gap} makes this value strictly smaller than $Y_{t-d}$.
\end{proof}

\begin{theorem}[mixed weighted-path rigidity]
\label{thm:mixed-path}
Let $L\subseteq\Npos$ and let $f\in\Mprim{m}{n}{L}$ be unary.  Fix one finite
mixed derivation $\delta$ of $f$, and let $Q_\delta\subseteq L$ be the finite
set of labels occurring in it.  Suppose a fixed $d\ge0$ satisfies
\[
 f(Y_t)=Y_{t-d}
\]
for every sufficiently large $t$.  Then there exists a threshold
$T=T(\delta,d)$ such that for every $t\ge T$ and for \emph{every} mixed selected
dependency chain $\chi$ of the concrete evaluation $\delta(Y_t)$:
\begin{enumerate}
\item $\chi$ begins at a formal input occurrence carrying the external unary input value $Y_t$;
\item if the descent-label word of $\chi$ is $q_1,\ldots,q_r$, in its actual path
order and with repetition, then
\[
 \boxed{\operatorname{wt}(\chi):=q_1+\cdots+q_r=d},
 \qquad q_i\in Q_\delta.
\]
\end{enumerate}
In particular,
\[
 d\in\gen{Q_\delta}\subseteq\gen{L}.
\]
Thus the canonical displacement is not merely represented by one successful
mixed path: it is an eventual invariant of every selected occurrence path of the
fixed derivation.
\end{theorem}

\begin{proof}
Choose $T$ larger than the thresholds in the chain-length, constant-origin,
and zone-uniqueness lemmas for the fixed derivation $\delta$ and target offset $d$,
and also larger than the low-zone threshold when $Q_\delta\ne\varnothing$.  Fix $t\ge T$ and an arbitrary mixed selected dependency chain for
the evaluation at $Y_t$; Lemma~\ref{lem:mixed-path-lifting} guarantees that the
selection is a finite occurrence path.  By Lemma~\ref{lem:mixed-constant-origin},
the chain cannot begin at a fixed numeral.  Since the
derivation is unary, its origin is therefore a formal input occurrence carrying the external unary input value
$Y_t\in\mathcal Z_t[0]$.

Read the chain forward and stop, if necessary, immediately before the first
descent labelled $q$ whose current zone index is below $q$.  All earlier
descents are legal, so Theorem~\ref{thm:mixed-zone-valuation} applies to that
initial segment.  Lemma~\ref{lem:mixed-low-zone} says that such a first illegal
descent would make the final target $Y_{t-d}$ unreachable.  Hence no illegal
descent occurs, and the valuation applies to the entire chain.

If the encountered labels, in their actual path order and with repetition, are
$q_1,\ldots,q_r$, then one fixed polynomial $P^\star$ satisfies
\[
 Y_{t-d}\in
 \mathcal Z_{t-(q_1+\cdots+q_r),P^\star}(t).
\]
Zone uniqueness, Lemma~\ref{lem:zone-unique}, forces
$d=q_1+\cdots+q_r$.  Every encountered label belongs to the finite derivation
support $Q_\delta$.  The threshold $T$ was chosen independently of the selected
chain, so arbitrariness of the chain proves the uniform ``every path'' statement
and hence path rigidity.
\end{proof}

\begin{remark}[path-weight conservation law]
The proof has no cancellation step.  Lower-basis dependencies may enlarge the
zone width but preserve its layer index; an actual descent labelled $q$ changes
that index by exactly $q$; and the two bottom barriers exclude the only regime in
which such an exact shift could leave the canonical scale.  Zone uniqueness then
turns equality of the final function value into equality of total path weight.
Thus canonical displacement is an additive conserved quantity of every eventual
selected occurrence path, even though the syntax may nest and mix labels
arbitrarily.  This is stronger than merely bounding the number or maximum size
of recursion steps.
\end{remark}

\begin{corollary}[derivation-local finite support]
\label{cor:finite-support}
Even when $L$ is infinite, every definable canonical descent
$\rho_{n,p}\in\Mprim{m}{n}{L}$ has a certificate
$p\in\gen{Q}$ for the finite label support $Q\subseteq L$ of one derivation.
\end{corollary}

\begin{proof}
Take a finite derivation of $\rho_{n,p}$ and apply
Theorem~\ref{thm:mixed-path} to the identity
$\rho_{n,p}(Y_t)=Y_{t-p}$ for $t\ge p$.
\end{proof}

\begin{theorem}[intrinsic stride spectrum of a primitive family]
\label{thm:primitive-spectrum}
Fix $n\ge2$ and $m<n$.  For every $L\subseteq\Npos$,
\[
 \boxed{\Str_n(\Mprim{m}{n}{L})=\gen{L}.}
\]
Equivalently, for every $p\ge1$,
\[
 \rho_{n,p}\in\Mprim{m}{n}{L}
 \quad\Longleftrightarrow\quad
 p\in\gen{L}.
\]
\end{theorem}

\begin{proof}
Suppose first that $p\in\gen{L}$.  Then
$p=l_1+\cdots+l_r$ for finitely many $l_i\in L$.  By
Lemma~\ref{lem:descent-internal}, each $\rho_{n,l_i}$ belongs to
$\Mprim{m}{n}{L}$; closure under composition and the canonical stride law give
\[
 \rho_{n,p}
 =\rho_{n,l_1}\circ\cdots\circ\rho_{n,l_r}
 \in\Mprim{m}{n}{L}.
\]
Conversely, if $\rho_{n,p}\in\Mprim{m}{n}{L}$, then
\eqref{eq:canonical-shift} gives
$\rho_{n,p}(Y_t)=Y_{t-p}$ for all $t\ge p$.  Theorem~\ref{thm:mixed-path} yields the stronger derivation-local certificate
$p\in\gen{Q_\delta}\subseteq\gen{L}$.
\end{proof}

The reverse inclusion is derivation-local and path-rigid: every sufficiently high selected path through a fixed derivation has the same total stride weight, namely the canonical displacement.

For a principal primitive family $L=\{q\}$, Theorem~\ref{thm:primitive-spectrum}
recovers the fixed-stride condition $q\mid p$.  Definability of a descent does
not by itself imply closure under recursion along that descent, a distinction
used below.

\begin{theorem}[exact canonical-descent membership and recovery]
\label{thm:descent-membership}
Fix $n\ge2$ and $m<n$.  For every additive submonoid
$\Gamma\le(\Nat,+)$ and every $p\ge1$,
\[
 \boxed{\rho_{n,p}\in\Hmix{m}{n}{\Gamma}\Longleftrightarrow p\in\Gamma},
 \qquad
 \boxed{\Str_n(\Hmix{m}{n}{\Gamma})=\Gamma}.
\]
Thus the monoid index is an intrinsic invariant of the function class at the
fixed canonical row.
\end{theorem}

\begin{proof}
Apply Theorem~\ref{thm:primitive-spectrum} to
$L=\Gamma\setminus\{0\}$; then $\gen{L}=\Gamma$ (with the trivial
$\Gamma=\{0\}$ case included).  The spectrum identity follows from the
definition.
\end{proof}

\begin{definition}[operator stride spectrum]
\label{def:operator-spectrum}
Let $C$ contain zero and all projections.  Define
\[
 \OpStr_n(C)=\{0\}\cup\{p\ge1:
 C\text{ is closed under }\SR_{\rho_{n,p}}
 \text{ on valid data from }C\}.
\]
Thus $\Str_n(C)$ is extensional---it asks whether the descent itself is a
function of $C$---whereas $\OpStr_n(C)$ is operational and asks whether that
descent is available as a closure scheme throughout $C$.
\end{definition}

\begin{proposition}[extensional--operational spectrum comparison]
\label{prop:operator-spectrum}
For every class $C$ as above,
\[
 \OpStr_n(C)\subseteq\Str_n(C).
\]
For the primitive-family and monoid-saturated classes below the basis threshold,
\[
 \{0\}\cup L\subseteq\OpStr_n(\Mprim{m}{n}{L})
 \subseteq\Str_n(\Mprim{m}{n}{L})=\gen{L},
\]
and, for every additive submonoid $\Gamma\le(\Nat,+)$,
\[
 \boxed{
 \OpStr_n(\Hmix{m}{n}{\Gamma})
 =\Str_n(\Hmix{m}{n}{\Gamma})
 =\Gamma.}
\]
Hence the index $\Gamma$ of a saturated mixed class is intrinsic both as a set
of definable canonical descents and as the exact set of canonical recursion
operators under which the class is closed.
\end{proposition}

\begin{proof}
If $p\in\OpStr_n(C)$, apply closure under $\SR_{\rho_{n,p}}$ to the base-zero,
current-address construction of Lemma~\ref{lem:descent-internal}; this yields
$\rho_{n,p}\in C$, proving $\OpStr_n(C)\subseteq\Str_n(C)$.  Every label in
$L$ is an available primitive recursion operator by definition of
$\Mprim{m}{n}{L}$, while Theorem~\ref{thm:primitive-spectrum} gives its descent
spectrum $\gen{L}$.  For $\Hmix{m}{n}{\Gamma}$, every $p\in\Gamma\setminus\{0\}$
is an available operator by definition, whereas
Theorem~\ref{thm:descent-membership} gives the reverse bound through the descent spectrum.  Thus both spectra equal $\Gamma$.
\end{proof}

Thus $\OpStr_n$ and $\Str_n$ separate operator availability from descent
definability; on saturated mixed classes both recover exactly $\Gamma$.

\section{Canonical descent monoids and mixed path factorization}
\label{sec:factorization}

Exact spectrum recovery determines which canonical descents occur.  The
weighted-path rigidity theorem yields more: it identifies the algebraic
factorizations of each displacement with the label-multisets that can occur on
selected dependency paths.

\begin{proposition}[canonical descent monoid]
\label{prop:descent-monoid}
For fixed $n\ge2$, $m<n$, and an additive submonoid
$\Gamma\le(\Nat,+)$, let
\[
 \mathcal D_n(\Gamma)=
 \{\rho_{n,p}:p\in\Gamma\}
 \subseteq \Hmix{m}{n}{\Gamma}.
\]
Under composition, $\mathcal D_n(\Gamma)$ is a commutative submonoid of the
unary-function monoid, and
\[
 \boxed{\Gamma\longrightarrow\mathcal D_n(\Gamma),\qquad
 p\longmapsto\rho_{n,p}}
\]
is a monoid isomorphism.  If $A=\{a_1,\ldots,a_k\}$ generates $\Gamma$, then
for all $\mathbf z,\mathbf w\in\Nat^k$,
\[
 \rho_{n,a_1}^{[z_1]}\circ\cdots\circ\rho_{n,a_k}^{[z_k]}
 =
 \rho_{n,a_1}^{[w_1]}\circ\cdots\circ\rho_{n,a_k}^{[w_k]}
 \quad\Longleftrightarrow\quad
 \sum_i a_i z_i=\sum_i a_i w_i.
\]
Moreover, the composition-indecomposable nonidentity elements of
$\mathcal D_n(\Gamma)$ are exactly
$\{\rho_{n,a}:a\in\Atoms(\Gamma)\}$; hence they form its unique minimal
composition-generating set.  Thus both the relations and the intrinsic
primitive generators of the canonical descent monoid are recovered inside the
function algebra.
\end{proposition}

\begin{proof}
Lemma~\ref{lem:stride-arithmetic}, together with $\rho_{n,0}=\operatorname{id}$,
gives $\rho_{n,p}\circ\rho_{n,q}=\rho_{n,p+q}$.  Surjectivity is by
definition.  If $p\ne q$, then for every sufficiently large $t$,
$\rho_{n,p}(Y_t)=Y_{t-p}\ne Y_{t-q}=\rho_{n,q}(Y_t)$, so the map is injective.
The relation criterion follows after reducing both sides by canonical stride
arithmetic to their total weights.  Finally,
$\rho_{n,p}=\rho_{n,u}\circ\rho_{n,v}$ with $u,v>0$ holds exactly when
$p=u+v$ in $\Gamma$.  Thus $\rho_{n,p}$ is composition-indecomposable
precisely when $p$ is an atom of $\Gamma$; uniqueness of the minimal atom
basis gives the final claim.
\end{proof}

The path theorem is naturally presentation-level.  For any (possibly infinite)
primitive label set $L\subseteq\Npos$, let $\Nat^{(L)}$ denote the finitely
supported maps $\mathbf z:L\to\Nat$, and define
\[
 \omega_L(\mathbf z)=\sum_{q\in L}q\,\mathbf z(q),
 \qquad
 Z_L(p)=\{\mathbf z\in\Nat^{(L)}:\omega_L(\mathbf z)=p\}.
\]
A vector $\mathbf z$ records the multiplicity of every primitive label on a
path; finite support is automatic for every finite derivation.

Equivalently, for finitely supported label vectors
$\mathbf z\in\Nat^{(L)}$ and
$\omega_L(\mathbf z)=\sum_{q\in L}q\mathbf z(q)$, the map
$\mathbf z\mapsto\rho_{n,\omega_L(\mathbf z)}$ induces
\[
 \Nat^{(L)}/\!\equiv_L\ \cong\ \mathcal D_n(\gen{L}),
 \qquad
 \mathbf z\equiv_L\mathbf w\Longleftrightarrow
 \omega_L(\mathbf z)=\omega_L(\mathbf w).
\]
Thus the only functional relations among primitive canonical descents are their
additive weight relations.

\begin{definition}[primitive-label path-factorization fibre]
\label{def:path-factorization}
For $p\in\gen{L}$, let $\PF^{m,n}_L(p)$ be the set of
$\mathbf z\in\Nat^{(L)}$ for which there exists a finite derivation of
$\rho_{n,p}$ in $\Mprim{m}{n}{L}$ and, for arbitrarily large canonical inputs
$Y_t$, a mixed selected dependency chain whose descent-label word has
multiplicity vector $\mathbf z$.
\end{definition}

The use of arbitrarily large $t$ removes finite initial accidents and makes the
definition insensitive to the threshold in Theorem~\ref{thm:mixed-path}.
The substantive direction below is the exclusion statement
$\PF^{m,n}_L(p)\subseteq Z_L(p)$: arbitrary nesting cannot create a selected
path signature of the wrong total weight.  The reverse inclusion is the direct
composition realization of an additive factorization.

\begin{theorem}[mixed path-factorization theorem]
\label{thm:path-factorization-general}
Fix $n\ge2$ and $m<n$.  For every $L\subseteq\Npos$ and every
$p\in\gen{L}$,
\[
 \boxed{\PF^{m,n}_L(p)=Z_L(p).}
\]
Thus the weighted-path theorem recovers not only the total mixed spectrum
$\gen{L}$ but the complete finite factorization fibre of every canonical
displacement relative to the primitive labels actually supplied.
\end{theorem}

\begin{proof}
The case $p=0$ is immediate from the empty factorization and the identity
derivation.  Assume $p>0$.  If $\mathbf z\in\PF^{m,n}_L(p)$, choose a witnessing
derivation and a canonical input beyond its rigidity threshold.
Theorem~\ref{thm:mixed-path} says that every selected path computing
$\rho_{n,p}(Y_t)=Y_{t-p}$ has total stride weight $p$.  Therefore
$\omega_L(\mathbf z)=p$, so $\mathbf z\in Z_L(p)$.

Conversely, let $\mathbf z\in Z_L(p)$.  Its support is finite.  By
Lemma~\ref{lem:descent-internal}, each $\rho_{n,q}$ with
$q\in\operatorname{supp}(\mathbf z)$ has a one-operator derivation in
$\Mprim{m}{n}{L}$.  Compose $\mathbf z(q)$ copies for every such $q$.
Canonical stride arithmetic gives the resulting function as $\rho_{n,p}$, and
on every sufficiently high $Y_t$ its obvious selected chain has multiplicity
vector $\mathbf z$.  Hence $\mathbf z\in\PF^{m,n}_L(p)$.
\end{proof}

The theorem becomes intrinsic after passing to the unique atom set
$A=\Atoms(\Gamma)=\{a_1<\cdots<a_k\}$ of a nonzero recovered monoid.  Put
\[
 Z_\Gamma(p)=\{\mathbf z\in\Nat^k:a_1z_1+\cdots+a_kz_k=p\},
 \qquad
 \Len_\Gamma(p)=\{|\mathbf z|:\mathbf z\in Z_\Gamma(p)\}.
\]

\begin{corollary}[intrinsic atomic factorization consequences]
\label{cor:atomic-consequences}
Let $A=\{a_1<\cdots<a_k\}=\Atoms(\Gamma)$.  For every $p\in\Gamma$,
\[
 \PF^{m,n}_{A}(p)=Z_\Gamma(p),\qquad
 \{|\mathbf z|:\mathbf z\in\PF^{m,n}_{A}(p)\}=\Len_\Gamma(p).
\]
Thus minimum and maximum factorization lengths are the corresponding atomic
selected-path descent counts.  Moreover $\Gamma$ is principal iff every
canonical descent has a unique atomic path Parikh vector, and the global path
elasticity is $a_k/a_1$.
\end{corollary}

\begin{proof}
The fibre identities are Theorem~\ref{thm:path-factorization-general}.  One
atom gives a principal uniquely factorizing monoid; two distinct atoms
$a<b$ give two factorizations of $\operatorname{lcm}(a,b)$.  Finally
$a_1|\mathbf z|\le p\le a_k|\mathbf z|$, with equality in the global ratio at
$\operatorname{lcm}(a_1,a_k)$.
\end{proof}

The fibre theorem also turns the standard finite presentations of numerical
semigroups into a finite rewrite calculus for mixed Step Recursion paths.  Let
$A=\{a_1,\ldots,a_k\}=\Atoms(\Gamma)$ and write
\[
 \pi_A:\Nat^k\to\Gamma,
 \qquad
 \pi_A(\mathbf z)=a_1z_1+\cdots+a_kz_k.
\]
A finite set $R\subseteq\ker(\pi_A)\subseteq\Nat^k\times\Nat^k$ is a
presentation when the congruence it generates is the whole kernel congruence.
Such finite presentations exist; after dividing by the gcd this is the usual
presentation theory of numerical semigroups
\cite[Chapter~7]{RosalesGarciaSanchez2009}.

\begin{theorem}[finite path-rewrite basis]
\label{thm:path-rewrite-basis}
Fix $n\ge2$, $m<n$, and a nonzero spectrum $\Gamma$ with atom set $A$.
There exists a finite relation set
$R_\Gamma\subseteq\Nat^k\times\Nat^k$ such that for every $p\in\Gamma$ and
every two atomic selected-path signatures
$\mathbf z,\mathbf w\in\PF_A^{m,n}(p)$ there is a finite chain
\[
 \mathbf z=\mathbf z_0,\mathbf z_1,\ldots,\mathbf z_s=\mathbf w
\]
inside $\PF_A^{m,n}(p)$ in which each step replaces a translated copy of one
relation from $R_\Gamma$: for some
$(\mathbf u,\mathbf v)\in R_\Gamma$ and $\mathbf c\in\Nat^k$,
\[
 \{\mathbf z_j,\mathbf z_{j+1}\}
 =\{\mathbf c+\mathbf u,\mathbf c+\mathbf v\}.
\]
The set $R_\Gamma$ may be chosen minimal, with every primitive relation
supported over a Betti element of the normalized numerical semigroup.
Consequently all atomic path-signature ambiguity in the entire mixed class is
generated by finitely many canonical ``Betti descents.''  This is a rewrite basis
for realizable multiplicity signatures across derivations; it does not assert an
in-place rewrite of one fixed concrete evaluation tree.
\end{theorem}

\begin{proof}
Choose a finite presentation $R_\Gamma$ of the factorization congruence
$\ker(\pi_A)$.  By Theorem~\ref{thm:path-factorization-general},
$\PF_A^{m,n}(p)=\pi_A^{-1}(p)$.  Hence any two selected-path signatures in the
same fibre are congruent modulo $R_\Gamma$, which by definition means they are
connected by a finite sequence of translated elementary relations as displayed
above.  Every elementary relation preserves $\pi_A$, so every intermediate
vector remains in the same factorization fibre and therefore, again by
Theorem~\ref{thm:path-factorization-general}, is itself realized by a selected
dependency path computing $\rho_{n,p}$.  Minimal presentations of numerical
semigroups can be chosen from relations joining components of the factorization
graphs of their Betti elements; scaling by the gcd does not change the
factorization congruence.  This gives the final statement.
\end{proof}

For $\Gamma=\gen{2,3}$, the unique minimal relation is
\[
 (3,0)\equiv(0,2),
\]
supported at the Betti descent $6$.  Thus every atomic path ambiguity in this
mixed class is generated by the local replacement
$2+2+2\leftrightarrow3+3$ inside a larger path.  In particular,
$\rho_{n,6}$ has atomic selected paths of lengths $3$ and $2$.

\section{Exact classification and transported one-dimensional geometry}

The classification theorem below is computational: it identifies class
inclusion exactly with inclusion of recovered stride monoids.  Once that
identification is established, several one-dimensional lattice statements are
standard consequences of finite generation and the classical theory of
numerical-semigroup oversemigroups
\cite{RosalesGarciaSanchez2009,RosalesOversemigroups2003}.  We include them to
state precisely what they mean for Step Recursion, not as new numerical-semigroup
theorems.

\begin{theorem}[classification and finite certificates]
\label{thm:classification}
Fix $n\ge2$ and $m<n$.  For additive submonoids
$\Gamma,\Lambda\le\Nat$ and $A=\Atoms(\Gamma)$,
\[
 \boxed{
 \Hmix{m}{n}{\Gamma}\subseteq\Hmix{m}{n}{\Lambda}
 \Longleftrightarrow \Gamma\subseteq\Lambda
 \Longleftrightarrow A\subseteq\Lambda.}
\]
Thus $\Gamma\mapsto\Hmix{m}{n}{\Gamma}$ is an order isomorphism, hence a
complete-lattice isomorphism, from $\Sub(\Nat)$ onto the saturated mixed
sector.  In particular,
\[
 \Hmix{m}{n}{\Gamma}\wedge_{\rm mix}\Hmix{m}{n}{\Lambda}
 =\Hmix{m}{n}{\Gamma\cap\Lambda},\qquad
 \Hmix{m}{n}{\Gamma}\vee_{\rm mix}\Hmix{m}{n}{\Lambda}
 =\Hmix{m}{n}{\gen{\Gamma\cup\Lambda}},
\]
and
\[
 \Hmix{m}{n}{p\Nat}\subseteq\Hmix{m}{n}{q\Nat}
 \Longleftrightarrow q\mid p.
\]
If inclusion fails, one atom $a\in A\setminus\Lambda$ gives the canonical
witness
$\rho_{n,a}\in\Hmix{m}{n}{\Gamma}\setminus\Hmix{m}{n}{\Lambda}$.
\end{theorem}

\begin{proof}
If $\Gamma\subseteq\Lambda$, every recursion operator used to generate
$\Hmix{m}{n}{\Gamma}$ is also available in $\Hmix{m}{n}{\Lambda}$, so
$\Hmix{m}{n}{\Gamma}\subseteq\Hmix{m}{n}{\Lambda}$.  Conversely, if the
class inclusion holds, then each $\rho_{n,p}$ with $p\in\Gamma$ belongs to
the target class, and Theorem~\ref{thm:descent-membership} yields
$p\in\Lambda$.  Since $\Gamma=\gen A$, the atomic criterion follows;
lattice operations and principal divisibility transport from $\Sub(\Nat)$.
\end{proof}

\begin{theorem}[intrinsic principal skeleton and finiteness]
\label{thm:join-irreducible}
\label{thm:acc}
At fixed $m<n$, the nonzero join-irreducible classes of the saturated mixed
sector are exactly
\[
 \boxed{\Hmix{m}{n}{d\Nat}\qquad(d\ge1).}
\]
Moreover:
\begin{enumerate}
\item every saturated mixed class is compact, the sector is countably infinite,
and it satisfies ACC;
\item it is not Artinian: for $S_r=\gen{2,2r+1}$,
\[
 \Hmix{m}{n}{S_1}\supsetneq\Hmix{m}{n}{S_2}
 \supsetneq\Hmix{m}{n}{S_3}\supsetneq\cdots .
\]
\end{enumerate}
Thus the principal saturated classes form the lattice-theoretically intrinsic
join-irreducible skeleton of a complete, countable, Noetherian mixed sector.
\end{theorem}

\begin{proof}
Transport the assertions through the lattice isomorphism of
Theorem~\ref{thm:classification}.  A nonzero submonoid of $\Nat$ with
at least two atoms is the join of two proper atom-generated submonoids, whereas
$d\Nat$ is join-irreducible: if a join has least positive element $d$, then one
of its factors already contains $d$.  Hence the nonzero join-irreducibles are
exactly the principal monoids $d\Nat$.  Finite atom bases give countability and
compactness.  For an ascending chain, its union is a submonoid of $\Nat$ and
therefore has a finite generating set; that finite set is contained in one
sufficiently late term, so the chain stabilizes.  Finally the monoids
$S_r=\gen{2,2r+1}$ form a strictly descending chain, and exact class recovery
transfers strictness to the displayed classes.
\end{proof}

\begin{remark}[primitive one-label versus saturated principal classes]
\label{rem:primitive-saturated-distinction}
For $d\ge1$ one always has
\[
 \Mprim{m}{n}{\{d\}}\subseteq\Hmix{m}{n}{d\Nat},
 \qquad
 \Str_n\!\left(\Mprim{m}{n}{\{d\}}\right)
 =\Str_n\!\left(\Hmix{m}{n}{d\Nat}\right)=d\Nat.
\]
The first class is the one-label fixed-stride algebra denoted
$\mathcal H^m_{n,d}$ in \cite{Osipov2026Step}; the second is, by definition,
closed under every recursion operator whose stride is a positive multiple of
$d$.  We do not identify these two full function classes here.  Such an
identification requires an operator-simulation theorem showing closure under
$\SR_{\rho_{n,kd}}$ from closure under $\SR_{\rho_{n,d}}$; it does not follow
merely from the definability of
$\rho_{n,kd}=\rho_{n,d}^{[k]}$.  None of the spectrum, path-factorization, or
saturated-lattice results in this paper uses that additional identification.
\end{remark}

\begin{theorem}[transported scale--defect and interval geometry]
\label{thm:principal-envelope}
\label{thm:graded-intervals}
\label{cor:hasse-components}
Let $\Gamma\ne\{0\}$ be an additive submonoid of $\Nat$ and write
\[
 d=\gcd(\Gamma\setminus\{0\}),\qquad \Gamma=dS.
\]
Then $S$ is a numerical semigroup, uniquely determined, and $d\Nat$ is the
least principal monoid containing $\Gamma$.  Its defect is finite:
$d\Nat\setminus\Gamma$ is finite, and if $c(S)$ is the conductor, then
\[
 p\ge d\,c(S)\quad\Longrightarrow\quad
 \boxed{p\in\Gamma\Longleftrightarrow d\mid p}.
\]
Moreover, for $\{0\}\ne\Gamma\subseteq\Lambda$:
\begin{enumerate}
\item $\gcd\Gamma=\gcd\Lambda$ iff $\Lambda\setminus\Gamma$ is finite iff the
interval $[\Gamma,\Lambda]$ is finite;
\item in that case the interval is graded, every cover $M\lessdot N$ adds
exactly one element, and every maximal chain has length
$\boxed{|\Lambda\setminus\Gamma|}$;
\item if the gcds differ, $[\Gamma,\Lambda]$ is infinite.
\end{enumerate}
Consequently the connected components of the nonzero Hasse graph are exactly
the gcd-scales.  The component of $\Gamma$ has unique maximum $d\Nat$, and
\[
 \boxed{
 \delta(\Gamma):=|d\Nat\setminus\Gamma|
 =\operatorname{dist}_{\rm Hasse}(\Gamma,d\Nat).}
\]
All statements transfer verbatim to the corresponding saturated Step Recursion
classes.  Thus $d$ is the asymptotic stride scale and $\delta$ the exact number
of one-descent expressive extensions needed to reach its principal envelope.
\end{theorem}

\begin{proof}
Writing $\Gamma=dS$ gives a numerical semigroup $S$, hence cofiniteness and the
conductor statement; a principal $q\Nat$ contains $\Gamma$ exactly when
$q\mid d$, proving minimality of $d\Nat$.  If
$\Gamma\subseteq\Lambda$ have the same gcd, both are cofinite in $d\Nat$, so
the difference and interval are finite.  For $M\subsetneq N$ in such an
interval, $x=\max(N\setminus M)$ can be adjoined alone: $x+u,2x\in M$ for
every positive $u\in M$.  Thus covers add one element and maximal chains have
length $|\Lambda\setminus\Gamma|$.  If the gcds are $a=br>b$, eventual
cofiniteness of $\Lambda$ supplies infinitely many
$x_k=b(rk+1)\in\Lambda\setminus\Gamma$ whose monoids
$\gen{\Gamma\cup\{x_k\}}$ are distinct intermediates.  Hence covers preserve
gcd and maximal-gap adjunction gives the Hasse-distance formula.  Transfer to
Step Recursion uses Theorem~\ref{thm:classification}.
\end{proof}

\section{Parallel clocks and the rank-two phase transition}
\label{sec:parallel-clocks}

The one-clock weighted-path and spectrum results are the core of the paper.  As
a secondary extension, we now test which parts of that geometry survive when
several canonical clocks are kept jointly.  The preceding classification uses
one canonical clock, so its spectrum lies in $\Nat$.  Its finite-basis and
Noetherian properties are special to that one-dimensional ambient monoid.  This
section isolates the higher-rank phenomenon at the already intrinsic
canonical-descent layer; no new recursion scheme is assumed yet.

For $r\ge1$ and $\mathbf p=(p_1,\ldots,p_r)\in\Nat^r$, define the parallel
canonical descent
\[
 \boldsymbol\rho^{(r)}_{n,\mathbf p}(y_1,\ldots,y_r)
 :=\bigl(\rho_{n,p_1}(y_1),\ldots,\rho_{n,p_r}(y_r)\bigr),
\]
using $\rho_{n,0}=\operatorname{id}_{\Nat}$.  For an additive submonoid
$\Gamma\le(\Nat^r,+)$ put
\[
 \mathcal D^{(r)}_n(\Gamma)
 :=\{\boldsymbol\rho^{(r)}_{n,\mathbf p}:\mathbf p\in\Gamma\}.
\]
These are monoids of canonical address transformations under composition.

\begin{theorem}[higher-rank canonical-descent representation]
\label{thm:vector-descent-representation}
For every $r\ge1$ and $n\ge2$, the map
\[
 \boxed{\mathbf p\longmapsto\boldsymbol\rho^{(r)}_{n,\mathbf p}}
\]
is an injective monoid homomorphism from $(\Nat^r,+)$ into
$\operatorname{End}(\Nat^r)$.  Hence, for every
$\Gamma\le\Nat^r$,
\[
 \boxed{\Gamma\cong\mathcal D^{(r)}_n(\Gamma)}.
\]
Moreover
\[
 \Gamma\longmapsto\mathcal D^{(r)}_n(\Gamma)
\]
is a complete-lattice isomorphism from $\Sub(\Nat^r)$ onto the lattice of
parallel canonical-descent algebras.
\end{theorem}

\begin{proof}
Coordinatewise stride arithmetic gives
\[
 \boldsymbol\rho^{(r)}_{n,\mathbf p}
 \circ\boldsymbol\rho^{(r)}_{n,\mathbf q}
 =\boldsymbol\rho^{(r)}_{n,\mathbf p+\mathbf q}.
\]
If $\mathbf p\ne\mathbf q$, choose a coordinate $i$ with $p_i\ne q_i$ and
$t\ge\max\{p_i,q_i\}$.  On a vector whose $i$th coordinate is the canonical
point $Y_t$, the two maps have distinct $i$th outputs
$Y_{t-p_i}$ and $Y_{t-q_i}$.  Thus the representation is injective.
Intersections are preserved literally, and the join of a family of submonoids
is the submonoid generated by their union; the same identity holds after the
homomorphic embedding.  Hence arbitrary meets and joins are preserved.
\end{proof}

In rank one Theorem~\ref{thm:classification} upgrades this descent
representation to a function-class classification; in higher rank the descent
algebra itself already changes character.

\begin{theorem}[rank-one/rank-two phase transition]
\label{thm:rank-two-transition}
Let $\mathfrak D_{n,r}$ be the lattice of parallel canonical-descent algebras
$\mathcal D^{(r)}_n(\Gamma)$.
\begin{enumerate}
\item For $r=1$, every element is finitely generated and compact,
$\mathfrak D_{n,1}$ is countably infinite, and it satisfies ACC.
\item For every $r\ge2$, the powerset order
$(\mathcal P(\Npos),\subseteq)$ order-embeds into $\mathfrak D_{n,r}$.
Consequently $\mathfrak D_{n,r}$ has cardinality $2^{\aleph_0}$, contains a
continuum antichain and an isomorphic copy of every countable partial order,
and admits infinite strict ascending chains; in particular ACC fails.
\item For every finite $r\ge1$, the nonzero join-irreducible elements of
$\mathfrak D_{n,r}$ are exactly the cyclic rays
\[
 \boxed{\mathcal D^{(r)}_n(\Nat\mathbf p)}
 \qquad(\mathbf p\in\Nat^r\setminus\{\mathbf0\}).
\]
\item An element $\mathcal D^{(r)}_n(\Gamma)$ is compact if and only if
$\Gamma$ is finitely generated.  Consequently, for $r\ge2$ only countably
many elements are compact while continuum many are noncompact.
\end{enumerate}
\end{theorem}

\begin{proof}
The rank-one claims were proved above.  For $r\ge2$ work in the first two
coordinates and set
$\Gamma_A=\gen{\{(1,a):a\in A\}}$ for $A\subseteq\Npos$.  Since a sum of
two nonzero elements has first coordinate at least $2$,
\[
 \Gamma_A\cap(\{1\}\times\Nat)=\{(1,a):a\in A\},
 \qquad A\subseteq B\Longleftrightarrow\Gamma_A\subseteq\Gamma_B.
\]
This embeds the powerset order, yielding all cardinality and order-complexity
claims.  Taking $A=\Npos$ gives infinitely many atoms $(1,a)$ and the strict
ascending finite truncations, so ACC fails and noncompact elements exist.
Every nonzero submonoid of $\Nat^r$ is atomic by descent of coordinate sum;
it is join-irreducible exactly when it has one atom, hence is a cyclic ray.
Finally, finite generation implies compactness by finite support in a join;
conversely $\Gamma$ is the directed join of its finitely generated
submonoids, so compactness implies finite generation.  Only countably many
finite subsets of $\Nat^r$ exist.
\end{proof}

\subsection*{Faithful coupled scalarization}

Scalarity itself need not lose synchronization.  Fix distinct primes
$\pi_1,\ldots,\pi_r$ and encode
\[
 E_r(\mathbf y)=\prod_{i=1}^r\pi_i^{y_i}-1,
 \qquad \mathcal C_r=E_r(\Nat^r).
\]
Unique factorization makes $E_r$ injective and $E_r(\mathbf0)=0$.  For
$\mathbf p\in\Nat^r$ define a total unary map
$\widehat\rho^{(r)}_{n,\mathbf p}$ by
\[
 \widehat\rho^{(r)}_{n,\mathbf p}(E_r(\mathbf y))
 =E_r(\boldsymbol\rho^{(r)}_{n,\mathbf p}(\mathbf y));
\]
off $\mathcal C_r$, let the zero label act as the identity and every nonzero
label map directly to $0$.

\begin{theorem}[faithful unary scalarization]
\label{thm:faithful-scalarization}
For finite $r\ge1$,
\[
 \boxed{\widehat\rho^{(r)}_{n,\mathbf p}\circ
 \widehat\rho^{(r)}_{n,\mathbf q}
 =\widehat\rho^{(r)}_{n,\mathbf p+\mathbf q}},
 \qquad
 \widehat\rho^{(r)}_{n,\mathbf0}=\operatorname{id},
\]
and $\mathbf p\mapsto\widehat\rho^{(r)}_{n,\mathbf p}$ is injective.  Hence
\[
 \Gamma\longmapsto
 \widehat{\mathcal D}^{(r)}_n(\Gamma)
 :=\{\widehat\rho^{(r)}_{n,\mathbf p}:\mathbf p\in\Gamma\}
\]
is a complete-lattice embedding of $\Sub(\Nat^r)$ into monoids of unary total
functions.  For strictly positive $\mathbf p$ these maps satisfy
$\widehat\rho^{(r)}_{n,\mathbf p}(0)=0$ and
$\widehat\rho^{(r)}_{n,\mathbf p}(x)<x$ for every $x>0$, so every orbit
reaches zero.  Moreover,
\[
 \boxed{\widehat\rho^{(r)}_{n,\mathbf p}\circ E_r
 =E_r\circ\boldsymbol\rho^{(r)}_{n,\mathbf p}.}
\]
Consequently synchronized recursion along $\mathbf p$ is, on the invariant
coded address space, semantically ordinary one-address step recursion along
$\widehat\rho^{(r)}_{n,\mathbf p}$ after decoding the predecessor code in the
transition.
\end{theorem}

\begin{proof}
On codes the composition identity is coordinatewise stride arithmetic; off
codes it follows from the identity/zero convention.  If
$\mathbf p\ne\mathbf q$, choose $i$ with $p_i\ne q_i$ and evaluate on a code
whose $i$th exponent is a sufficiently high canonical point $Y_t$ and whose
other exponents are zero.  The resulting $\pi_i$-exponents
$Y_{t-p_i}$ and $Y_{t-q_i}$ differ, proving injectivity.  Meets and joins then
transport exactly as in Theorem~\ref{thm:vector-descent-representation}.
For strictly positive $\mathbf p$, every positive exponent is strictly
decreased by its canonical coordinate descent while zero exponents remain
zero.  Hence every nonzero code is sent to a smaller code and repeated
iteration reaches $E_r(\mathbf0)=0$; noncodes reach $0$ in one step.  The
conjugacy identity is immediate from the definition; substituting
$z=E_r(\mathbf y)$ into the synchronized recursion equation gives the final
semantic statement.
\end{proof}

This is a semantic conjugacy statement only: it does not assert that the prime
encoding $E_r$, its decoding, or the total coded descents belong internally to
the Step Recursion class under study.

Two elementary obstructions show that faithfulness requires genuine coupling.

\begin{proposition}[additive-label and marginal no-go]
\label{prop:coupling-no-go}
For $r\ge2$:
\begin{enumerate}
\item no additive homomorphism $(\Nat^r,+)\to(\Nat,+)$ is injective;
\item even all proper coordinate marginals are not faithful.  The proper
submonoid
\[
 \Gamma_{\rm even}
 =\{\mathbf p:\sum_i p_i\equiv0\pmod2\}\subsetneq\Nat^r
\]
satisfies
\[
 \boxed{\pi_I(\Gamma_{\rm even})=\Nat^I
 \quad(\varnothing\ne I\subsetneq\{1,\ldots,r\}),}
\]
exactly as $\Nat^r$ does.
\end{enumerate}
Thus arbitrary synchronization requires genuinely $r$-way information, even
though Theorem~\ref{thm:faithful-scalarization} shows that one coupled unary
function channel is enough to carry it.
\end{proposition}

\begin{proof}
Write an additive map as
$\ell(\mathbf p)=\sum_i a_ip_i$.  A zero $a_i$ destroys injectivity; if all
$a_i>0$, the distinct vectors $a_j\mathbf e_i$ and $a_i\mathbf e_j$ have the
same image for $i\ne j$.  For the marginal claim, after prescribing any proper
set of coordinates, choose one omitted coordinate with the parity required to
make the total sum even and set the remaining omitted coordinates to zero.
\end{proof}

Combining Theorems~\ref{thm:rank-two-transition} and
\ref{thm:faithful-scalarization}, the entire rank-two explosion---continuum
many spectra and an embedded powerset order---is therefore realizable inside
monoids of unary total functions $\Nat\to\Nat$.  The boundary is sharper than
``vector versus scalar'': coupled scalarization can be faithful, while
additive labels and every hierarchy of proper coordinate marginals can still
erase the joint geometry.

\begin{remark}[from joint descent algebra to scalar function classes]
Theorem~\ref{thm:rank-two-transition} is an exact theorem about the canonical
descent algebra already present in Step Recursion.  Section~\ref{sec:multiclock-collapse}
asks what survives after one adds a natural synchronized multi-clock recursion
scheme but returns to ordinary scalar function classes.  The answer is not the
naive equality with an arbitrary submonoid of $\Nat^r$: canonical extensional
membership rectangularizes the joint spectrum.  This separates synchronized
descent geometry from its scalar extensional shadow.
\end{remark}

\section{Synchronized multi-clock recursion and extensional collapse}
\label{sec:multiclock-collapse}

The joint descent algebra retains vector labels.  We now ask what survives in
ordinary scalar function classes when the recursion address has several
synchronously advanced clocks.  The answer is exact but lossy: selected scalar
dependency paths see individual coordinate descents rather than the full
vector label.

Fix $r\ge2$.  For a strict positive vector
$\mathbf q=(q_1,\ldots,q_r)\in(\Npos)^r$, put
\[
 \boldsymbol\rho_{n,\mathbf q}(\mathbf y)
 =\bigl(\rho_{n,q_1}(y_1),\ldots,\rho_{n,q_r}(y_r)\bigr).
\]
Strict positivity ensures that repeated application reaches $\mathbf0$ from
every $\mathbf y\in\Nat^r$.  This should not be confused with classical
simultaneous recursion, which couples several output functions along one scalar
recursion variable.  Here the output remains scalar while the recursion address
itself has $r$ synchronously advanced clock coordinates.

\begin{definition}[synchronized multi-clock bounded step recursion]
Let $g:\Nat^k\to\Nat$, $h:\Nat^{k+r+1}\to\Nat$, and
$b:\Nat^{k+r}\to\Nat$ be earlier functions.  Bounded synchronized recursion
along $\mathbf q\in(\Npos)^r$ forms $f:\Nat^{k+r}\to\Nat$ by
\[
 f(\bar x,\mathbf0)=g(\bar x),
\]
\[
 f(\bar x,\mathbf y)
 =h\bigl(\bar x,\boldsymbol\rho_{n,\mathbf q}(\mathbf y),
              f(\bar x,\boldsymbol\rho_{n,\mathbf q}(\mathbf y))\bigr)
 \qquad(\mathbf y\ne\mathbf0),
\]
provided $f(\bar x,\mathbf y)\le b(\bar x,\mathbf y)$ everywhere.  Denote the
operation by $\VSR_{\mathbf q}$.  For
$L\subseteq(\Npos)^r$, let
\[
 \mathcal M^{m,(r)}_n[L]
 :=\Cl_{\circ,\{\VSR_{\mathbf q}:\mathbf q\in L\}}(B_m).
\]
Finally define the flattened scalar label set and its monoid by
\[
 \Flat(L)=\{q_i:\mathbf q=(q_1,\ldots,q_r)\in L,\ 1\le i\le r\},
 \qquad
 M(L)=\gen{\Flat(L)}\le\Nat.
\]
\end{definition}

The first obstruction is independent of the recursion scheme.

\begin{definition}[coordinatewise vector descent shadow]
If $C$ is a scalar function class containing zero and projections and closed
under composition, define
\[
 \VStr_{n,r}(C)
 :=\Bigl\{\mathbf p=(p_1,\ldots,p_r)\in\Nat^r:
       (\mathbf y\mapsto\rho_{n,p_i}(y_i))\in C
       \text{ for every }i\Bigr\},
\]
where $\rho_{n,0}=\operatorname{id}$.
\end{definition}

\begin{theorem}[rectangularity no-go theorem]
\label{thm:rectangularity}
For every such scalar function class $C$,
\[
 \boxed{\VStr_{n,r}(C)=\Str_n(C)^r.}
\]
Consequently the natural coordinatewise canonical vector spectrum of an
ordinary scalar function class is always rectangular.  In particular, no
nonrectangular submonoid $\Gamma\le\Nat^r$ can be recovered by this naive
extension of canonical-descent membership.
\end{theorem}

\begin{proof}
If $p\in\Str_n(C)$, composition of the unary descent $\rho_{n,p}$ with the
$i$th projection gives
$\mathbf y\mapsto\rho_{n,p}(y_i)$, so
$\Str_n(C)^r\subseteq\VStr_{n,r}(C)$.  Conversely, if the $i$th coordinate
map belongs to $C$, compose it with the tuple of arguments having a single
variable in position $i$ and zeros elsewhere.  This recovers the unary
function $\rho_{n,p_i}$, hence $p_i\in\Str_n(C)$.  Apply this to every
coordinate.
\end{proof}

The synchronized recursion scheme admits an exact scalar shadow theorem.  The
proof is the one-clock mixed valuation with one additional bookkeeping step:
a selected scalar dependency path sees one coordinate of a vector address at a
time, and therefore sees scalar labels $q_i$, not the full vector $\mathbf q$.

\begin{lemma}[multi-clock selected-path lifting]
\label{lem:multiclock-path-lifting}
Fix $m<n$ and a finite synchronized multi-clock derivation $\delta$ with
$r$ clock coordinates.  In every concrete evaluation of a scalar output of
$\delta$, the one-clock occurrence-selection rules extend by treating the
$r$ coordinates of each vector address as distinct scalar occurrences.  Every
selected scalar dependency path obtained in this way is finite.  At a
$\VSR_{\mathbf q}$ node, with $\mathbf q=(q_1,\ldots,q_r)$, lifting a selected
transition origin has exactly three forms:
\begin{enumerate}
\item a parameter origin lifts through the corresponding outer argument;
\item an origin at coordinate $i$ of the next vector address contributes the
scalar edge
\[
 y_i\longmapsto\rho_{n,q_i}(y_i),
\]
which is an equality when $y_i=0$ and otherwise is an active descent labelled
$q_i$;
\item a previous-state origin lifts to the preceding common vector stage.
\end{enumerate}
Repeated previous-state lifting strictly decreases that vector-stage index.
Moreover, if the derivation is unary and is evaluated at the canonical input
$Y_t$, then every selected scalar path has length bounded by one fixed
polynomial $P_\delta(t)$.
\end{lemma}

\begin{proof}
Use induction on construction rank, with a secondary induction on the common
vector-stage index at a synchronized recursion node.  Initial-function and
composition nodes are exactly as in Lemma~\ref{lem:mixed-path-lifting}; at a
composition node the selected head occurrence chooses one concrete argument
occurrence, so arbitrary duplications, permutations, and identifications of
the external unary variable are already covered.  At a $\VSR_{\mathbf q}$
node, a parameter or one coordinate of the next vector address lifts directly,
while a previous-state origin moves to the preceding common vector stage.
Hence stage unwinding terminates and no occurrence cycle is possible.

For the quantitative statement, at vector address $\mathbf y$ one has
\[
 D_{n,\mathbf q}(\mathbf y)
 =\max_i D_{n,q_i}(y_i)
 \le \max_i D_{n,1}(y_i).
\]
The internal-value induction of Lemma~\ref{lem:mixed-internal-envelope}
applies coordinatewise with the scalar cap $N=\max_i y_i$; fixed $r$ changes
only derivation-dependent constants.  Thus on root input $Y_t$ every
node-specific vector recursion depth is $O_\delta(t+1)$.  The same
construction-rank recurrence as in Lemma~\ref{lem:mixed-chain-length} then
bounds every selected scalar path polynomially in $t$.  The argument is
uniform over all choices of selected occurrences.
\end{proof}

\begin{lemma}[flattening of multi-clock selected paths]
\label{lem:multiclock-flattening}
Fix $m<n$, a finite synchronized multi-clock derivation $\delta$, and let
$Q_\delta\subseteq(\Npos)^r$ be the finite set of vector labels occurring in
it.  Put
\[
 Q_\delta^\flat
 =\{q_i:\mathbf q\in Q_\delta,\ 1\le i\le r\}.
\]
If a unary function $f$ computed by $\delta$ satisfies
\[
 f(Y_t)=Y_{t-d}
\]
for every sufficiently large $t$, then every sufficiently high selected scalar
dependency path has flattened stride weight $d$; in particular
\[
 \boxed{d\in\gen{Q_\delta^\flat}.}
\]
\end{lemma}

\begin{proof}
Fix a concrete evaluation of the unary derivation $\delta$ at the root input
$Y_t$.  By Lemma~\ref{lem:multiclock-path-lifting}, every selected scalar path
is finite, has polynomial length in $t$, and its only active canonical descent
edges have labels in
\[
 Q_\delta^\flat
 =\{q_i:\mathbf q\in Q_\delta,\ 1\le i\le r\}.
\]
All remaining numerical edges are equalities or the same oriented lower-basis
edges used in the one-clock proof.
The two one-clock barrier arguments also survive this extension.  For a fixed
numeral origin, replace every active coordinate descent by the identity.  The
polynomial selected-path bound just proved and
Lemma~\ref{lem:mixed-basis-stability} give an upper envelope
$F^{[R(t)]}(c+R(t))$ for fixed $c$ and polynomial $R$; the canonical-gap lemma
places this strictly below $Y_{t-d}$ for large $t$.  Hence every sufficiently
high selected path reaching $Y_{t-d}$ originates at the unique external root
value $Y_t$.  Likewise, if an active $q_i$-edge were first encountered while
the current scalar zone index were $s<q_i$, then after that edge the same
identity-replacement upper-envelope argument used in
Lemma~\ref{lem:mixed-low-zone} would keep the remainder of the path strictly
below $Y_{t-d}$.  Thus no such illegal active edge occurs on a sufficiently
high selected path reaching the target.

We may therefore propagate the scalar zone index along the whole path.  Start
with $Y_t\in\mathcal Z_t[0]$ and induct over its numerical edges.  Equality
leaves index and width unchanged; an oriented lower-basis edge preserves the
index and increases width by at most the derivation constant from
Lemma~\ref{lem:mixed-basis-stability}; and an active coordinate descent
labelled $q_i$ changes the index from $s$ to $s-q_i$ without changing width by
Lemma~\ref{lem:mixed-descent-displacement}.  Since the total number of edges is
polynomial in $t$, all accumulated width is bounded by one fixed polynomial
$P^\star$.  If $W$ is the sum of the active coordinate labels on the selected
path, its endpoint therefore satisfies
\[
 Y_{t-d}\in\mathcal Z_{t-W,P^\star}(t).
\]
Zone uniqueness, Lemma~\ref{lem:zone-unique}, forces $W=d$.  Writing the active
labels in path order as $s_1,\ldots,s_v$ gives
\[
 d=s_1+\cdots+s_v,
 \qquad s_j\in Q_\delta^\flat,
\]
and hence $d\in\gen{Q_\delta^\flat}$.
\end{proof}

\begin{theorem}[exact synchronized multi-clock shadow]
\label{thm:multiclock-shadow}
Fix $n\ge2$, $m<n$, $r\ge2$, and
$L\subseteq(\Npos)^r$.  Then
\[
 \boxed{\Str_n\bigl(\mathcal M^{m,(r)}_n[L]\bigr)=M(L)}
\]
and therefore
\[
 \boxed{
 \VStr_{n,r}\bigl(\mathcal M^{m,(r)}_n[L]\bigr)=M(L)^r.}
\]
Thus synchronized vector labels are rectangularized by ordinary scalar
extensional closure.
\end{theorem}

\begin{proof}
For the forward generation, fix $\mathbf q\in L$ and a coordinate $i$.
One synchronized recursion along $\mathbf q$ with base zero and transition
returning the $i$th coordinate of the next vector address defines
\[
 F_i(\mathbf y)=\rho_{n,q_i}(y_i).
\]
The projection $y_i$ is an admissible bound.  Substituting zero into all clock
slots except $i$ gives the unary descent $\rho_{n,q_i}$.  Hence every element
of $\Flat(L)$ lies in the scalar stride spectrum, and closure under composition
gives
$M(L)\subseteq\Str_n(\mathcal M^{m,(r)}_n[L])$.

Conversely, let $\rho_{n,p}$ belong to the class and fix one finite derivation.
Since $\rho_{n,p}(Y_t)=Y_{t-p}$ for all $t\ge p$,
Lemma~\ref{lem:multiclock-flattening} gives
$p\in\gen{Q_\delta^\flat}\subseteq M(L)$.  This proves the scalar identity.
The vector identity is now Theorem~\ref{thm:rectangularity}.
\end{proof}

For a joint descent monoid $\Gamma\le\Nat^r$, define its scalar flattening and
rectangular closure by
\[
 \Sc(\Gamma)
 :=\gen{\{p_i:\mathbf p\in\Gamma,\ 1\le i\le r\}},
 \qquad
 \Rect_r(\Gamma):=\Sc(\Gamma)^r.
\]
The next statement quantifies exactly how much joint information the natural
extensional vector shadow can lose.

\begin{theorem}[rectangular reflection and maximal fibres]
\label{thm:synchronization-collapse}
Fix finite $r\ge2$.  Let
\[
 P_r:\Sub(\Nat)\to\Sub(\Nat^r),\qquad P_r(M)=M^r.
\]
Then:
\begin{enumerate}
\item for every $\Gamma\le\Nat^r$ and $M\le\Nat$,
\[
 \boxed{\Sc(\Gamma)\subseteq M\quad\Longleftrightarrow\quad
 \Gamma\subseteq M^r.}
\]
Thus $\Sc$ is left adjoint to the order embedding $P_r$,
$\Rect_r=P_r\!\circ\Sc$ is the associated closure operator, and
$\Rect_r(\Gamma)$ is the least rectangular submonoid containing $\Gamma$.
Its fixed points are exactly the $M^r$.
\item If $L\subseteq(\Npos)^r$ and $\Gamma=\gen L$, then
\[
 \boxed{
 \VStr_{n,r}\bigl(\mathcal M^{m,(r)}_n[L]\bigr)=\Rect_r(\Gamma).}
\]
\item The zero fibre is trivial, but every nonzero fibre is maximally large:
for every $M\ne\{0\}$,
\[
 \boxed{\{\Gamma\le\Nat^r:\Rect_r(\Gamma)=M^r\}}
\]
contains an order-embedded copy of $(\mathcal P(\Nat),\subseteq)$.  Hence it
has cardinality $2^{\aleph_0}$ and contains continuum antichains.
\end{enumerate}
Moreover $\Sc$ preserves arbitrary joins, but for $r\ge2$ it need not preserve
meets.
\end{theorem}

\begin{proof}
The adjunction is immediate: $\Gamma\subseteq M^r$ iff every coordinate of
every element of $\Gamma$ lies in $M$, iff $\Sc(\Gamma)\subseteq M$.  This
gives the closure and fixed-point claims, while part~(2) is
Theorem~\ref{thm:multiclock-shadow}.  For a nonzero $M$, let $\mu$ be its
least positive element, $B=\Atoms(M)$, and choose infinite
$T\subseteq M\setminus(B\cup\{0\})$.  For $A\subseteq T$ put
\[
 \Gamma_A=\gen{\{(\mu,b):b\in B\}\cup\{(\mu,a):a\in A\}}
 \quad(r=2),
\]
and use $(\mu,x,\mu,\ldots,\mu)$ for $r>2$.  Then
$\Sc(\Gamma_A)=M$, while the first-coordinate-$\mu$ slice recovers $B\cup A$;
hence $A\subseteq A'$ iff $\Gamma_A\subseteq\Gamma_{A'}$.  Thus every
nonzero fibre contains the powerset order and has cardinality continuum.  The
zero fibre is trivial.  As a left adjoint $\Sc$ preserves joins; it does not
preserve meets, since $\gen{(1,2)}$ and $\gen{(2,1)}$ have scalar shadow
$\Nat$ but trivial intersection.
\end{proof}

\begin{remark}[sharpness of the collapse]
\label{cor:meet-collapse}
\label{ex:one-vector-collapse}
The loss can be complete.  With
$\Gamma=\gen{(1,2)}$ and $\Lambda=\gen{(2,1)}$,
\[
 \Rect_2(\Gamma)=\Rect_2(\Lambda)=\Nat^2,
 \qquad \Gamma\cap\Lambda=\{\mathbf0\}.
\]
Likewise one synchronized label $L=\{(1,2)\}$ generates the joint ray
$\{(k,2k):k\in\Nat\}$ but has ordinary vector shadow $\Nat^2$.  These
statements concern canonical descent membership, not equality of the full
scalar function classes: noncanonical functions may retain additional
synchronization information.
\end{remark}

\section{Discussion and conclusion}

The one-clock results are driven by a path-weight conservation law: canonical
displacement is recovered from every sufficiently high selected occurrence
path, despite arbitrary nesting of finitely many primitive labels.  The
spectrum theorem recovers which canonical descents occur, while the
path-factorization theorem recovers their complete additive synthesis
signatures and finite presentation-level rewrite bases.  Saturation makes the
same monoid an intrinsic invariant of class inclusion, with the principal
saturated classes $\Hmix{m}{n}{d\Nat}$ as the join-irreducible skeleton.

Higher rank separates scalarity from separability.  Joint descents, as address
maps, realize arbitrary $\Gamma\le\Nat^r$ and already in rank two have
continuum order complexity; coupled unary coding is an external conjugacy of
those maps, not an internal definability statement.  Ordinary canonical scalar
observation of synchronized multi-clock recursion instead sees only the
reflected rectangular shadow $\Rect_r(\Gamma)$, whose every nonzero fibre
contains the powerset order.
Thus the decisive issue is not scalar output but whether the observable keeps
joint address information.  In this sense the higher-rank results form an
obstruction theorem: separable canonical observation necessarily destroys all
nonrectangular synchronization geometry, while a coupled scalar coding need
not do so.

\section*{Acknowledgements}
The author acknowledges the use of large language models for copyediting,
grammatical correction, and language polishing.  The author reviewed the
resulting text and takes full responsibility for the final manuscript.

\end{document}